\documentclass{amsart}
\usepackage{amssymb}
\usepackage{amsmath}
\usepackage[dvips]{graphicx}
\usepackage{hyperref}
\usepackage{color}
\usepackage{tikz}
\usetikzlibrary{arrows}

\numberwithin{equation}{section}

\newcommand{\be}{\begin{equation}}
\newcommand{\ee}{\end{equation}}
\newcommand{\ds}{\displaystyle}

\renewcommand{\ss}{\scriptstyle}

\DeclareMathOperator{\sgn}{sgn}
\DeclareMathOperator{\pf}{Pf}
\DeclareMathOperator{\rev}{rev}
\newcommand{\cO}{\mathcal O}
\newmuskip\pFqmuskip

\newcommand*\pFq[6][8]{%
  \begingroup 
  \pFqmuskip=#1mu\relax
  \mathcode`\,=\string"8000
  \begingroup\lccode`\~=`\,
  \lowercase{\endgroup\let~}\pFqcomma
  {}_{#2}F_{#3}{\left[\genfrac..{0pt}{}{#4}{#5};#6\right]}%
  \endgroup
}
\newcommand{\pFqcomma}{\mskip\pFqmuskip}

\newtheorem{thm}{Theorem}[section]

\newtheorem{cor}[thm]{Corollary}
\newtheorem{lem}[thm]{Lemma}
\newtheorem{defn}[thm]{Definition}
\newtheorem{rem}[thm]{Remark}

\newtheorem{exmp}[thm]{Example}

\makeatletter
\def\Ddots{\mathinner{\mkern1mu\raise\p@
\vbox{\kern7\p@\hbox{.}}\mkern2mu
\raise4\p@\hbox{.}\mkern2mu\raise7\p@\hbox{.}\mkern1mu}}
\makeatother

\title[Current Loops and Magnetic Monopoles]{A Statistical Model of Current Loops and Magnetic Monopoles}
\author{Arvind Ayyer}
\address{Department of Mathematics, Indian Institute of Science, Bangalore  560012, India.}
\email{arvind@math.iisc.ernet.in}
\date{\today}

\begin{document}

\begin{abstract}
We formulate a natural model of loops and isolated vertices for arbitrary planar graphs, which we call the monopole-dimer model. We show that the partition function of this model can be expressed as a determinant. We then extend the method of Kasteleyn and Temperley-Fisher to calculate the partition function exactly in the case of rectangular grids. This partition function turns out to be a square of a polynomial with positive integer coefficients when the grid lengths are even. Finally, we analyse this formula in the infinite volume limit and show that the local monopole density, free energy and entropy can be expressed in terms of well-known elliptic functions.
Our technique is a novel determinantal formula for the partition function of a  model of isolated vertices and loops for arbitrary graphs. 
\end{abstract}

\maketitle

\section{Introduction} 

The dimer model on a planar graph $G$ is a statistical mechanical
model which idealises the adsorption of diatomic molecules on $G$. The
associated combinatorial problem is the weighted enumeration of all
dimer covers of $G$, also known as perfect matchings or 1-factors. 
This problem was
solved in a beautiful and explicit way by Kasteleyn
\cite{kasteleyn1961,kasteleyn1963} and by Temperley-Fisher
\cite{temperley1961,fisher1961}.

The monomer-dimer model on the other hand, which idealises the adsorption of both monoatomic as well as diatomic
molecules on $G$, has not had as
much success. In this case, one considers the weighted enumeration of all possible matchings of $G$ with separate fugacities for
both kinds of molecules.  Equivalently, this is the problem of counting all matchings of $G$. There is some indirect evidence that it is not likely to be
exactly solvable \cite{jerrum1987}.
It has been rigorously shown that the monomer-dimer model does not exhibit
phase transitions \cite{gruberkunz1971,heilmann1972}.
The only solutions so far are obtained by
perturbative expansions (see the review in \cite{heilmann1972}, for
example). The asymptotics of the free energy has been studied by various 
authors, see \cite{bondy1966,hammersley1966,hammersley1970} for instance. There have also been several numerical studies \cite{kenyon1996,kong2006a,kong2006b} as well as study of monomer correlations in a sea of dimers \cite{fisher1963}.
We note that there has been some success in solving restricted versions of the classical monomer-dimer model exactly, either for finite size or in the limit of infinite size.
Such is the case for a single monomer on the boundary \cite{tzeng2003,wu2006}, 
arbitrary monomers on the boundary in the scaling limit \cite{priezzhev2008}
and a single monomer in the bulk in the thermodynamic limit \cite{bouttier2007,poghosyan2008}. More recently, after the completion of this work, there has appeared a Grassmannian approach to computing the dimer model partition function with fixed locations of monomers exactly \cite{allegra2014}. On the hexagonal lattice, a lot of work has been done on monomer correlations by Ciucu, see \cite{ciucu2010} and references therein.

We note in passing that signed dimer models 
and signed loop models have gained attention in statistical physics recently, the former in the context of spin liquids \cite{ddr2012} and the latter as an approach towards solving the Ising model \cite{kager2013}.

In this article, we will consider a signed variant of the monomer-dimer model on any planar graph, which we call the monopole-dimer model. This model will turn out to be a natural generalisation of the well-known dimer model\footnote{more precisely, the double-dimer model}, also defined for any planar graph.
The configurations of this model are subgraphs consisting of isolated vertices, doubled edges and oriented loops of even length on the graph such that each vertex is attached to exactly zero or two edges. Each configuration can be thought of as a superposition of two monomer-dimer configurations with the same monomer locations.
The reason for the nomenclature will be explained in Section~\ref{sec:mondim}, when the weights associated to these configurations are specified.
We will prove that the partition function of the monopole-dimer model can be written as a determinant. 
This property is useful from a computational point of view because one can obtain a lot of information about the model using nothing more than basic linear algebra. This approach has been extremely fruitful in studying many models in statistical physics, such as the Ising model in one-dimension \cite{mccoy1973}, the sandpile model \cite{dhar1990} and
the dimer model for planar graphs \cite{kenyon1997}.

We will use this determinant formula to express the partition function of the monopole-dimer model on the two-dimensional grid as a product.
This will turn out to give a natural generalisation of Kasteleyn's and Temperley-Fisher's formula for the dimer model on the rectangular grid. 
It will turn out, for not obvious reasons, that the partition function will be an exact square when the sides of the rectangle are even. This is in contrast to the double-dimer model \cite{kasteleyn1961,fisher1961}, where the partition function is the determinant of an even anti-symmetric matrix, and hence is obviously the square of the corresponding Pfaffian.

We will then derive explicit formulas for the free energy of the monopole-dimer model in terms of known elliptic functions in the infinite size limit and compare it with existing results for the monomer-dimer model, both rigorous and numerical. We will also calculate the entropy and the monopole density.
The starting point, namely the determinant formula, is a consequence of a more general model of oriented loops, doubled edges and vertices
on a general graph, which we will first explain.

The plan of the paper is as follows.
We will first define a new loop-vertex model on arbitrary graphs in Section~\ref{sec:loopvert} and show that the partition function of the model can be written as a determinant in Theorem~\ref{thm:pfloopvertex}. We will then define the monopole-dimer model in Section~\ref{sec:mondim} and use results proved in the previous section to show that its partition function can also be written as a determinant in Theorem~\ref{thm:pfmonopoledimer}. We will then specialise to the two-dimensional grid graph in Section~\ref{sec:mondimgrid} and give an explicit product formula for the partition function in Theorem~\ref{thm:partfngrid}. We finally discuss the asymptotic limit of $\mathbb{Z}^2$ in Section~\ref{sec:asymp}.

The statements of the paper can be verified using the Maple program file 
\texttt{Monopo\-le.maple} available from the author's webpage or as an ancillary file from the \texttt{arXiv} source.

\section*{Acknowledgements}
We would like to acknowledge support in part by a UGC Centre for Advanced Study grant.
We would also like to thank C. Krattenthaler and J. Bouttier for discussions,
T. Amdeberhan for conjecturing \eqref{partfn-grid},
K. Damle and R. Rajesh for suggesting references, and 
M. Krishnapur for many helpful discussions. We also thank the anonymous referees for several useful comments.

\section{A Loop-Vertex Model on General Graphs} 
\label{sec:loopvert}
We begin by defining a model of isolated vertices and loops of even length on arbitrary graphs. The usefulness of the results here is that they are very general, and might be interesting in their own right. At this point, we do not know of any relevant physical situation where this model could be applied. Part of the objective of this section is to make the proof of the determinantal formula for the partition function of the monopole-dimer model simpler. The reader interested in the monopole-dimer model should feel free to skip this section.

Our input data is a simple (not necessarily planar), undirected
vertex- and edge-weighted labelled graph $G = [V,E]$ on $n$ vertices and an arbitrary assignment of arrows along each edge, called the orientation $\cO$ on $G$. We will denote vertex weights by $x(v)$ for $v \in V$ and edge weights as $a(v,v') \equiv a(v',v)$ whenever $(v,v') \in E$.
Any labelled graph comes with a canonical orientation, the one got by directing edges from a lower vertex to a higher one. 

\begin{defn} \label{def:loopvertconf}
A {\bf loop-vertex configuration} $C$ consists of a subgraph of $G$ of edges which form directed loops of even length including doubled edges (to be thought of as loops of length 2), with the property that every vertex belongs to exactly zero or two edges. Let $\mathcal{L}$ be the set of loop-vertex configurations.
\end{defn}

Note that the number of isolated vertices has the same parity as the size of the graph.
We first define the signed weight of a loop in $C$. First, the sign of an edge $(v_1, v_2)$, denoted $\sgn(v_1,v_2)$
is $+1$ if the orientation is from $v_1 \to v_2$ in $\cO$ and $-1$ otherwise. 
Then, given an even oriented loop $\ell = (v_1, \dots, v_{2 n}, v_1)$, the
weight of the loop is
\be \label{defloop}
w(\ell) = -\prod_{j=1}^{2 n} \sgn(v_j,v_{j+1}) \; a(v_j,v_{j+1}),
\ee
with the understanding that $v_{2 n+1} = v_1$. The reason for the
overall minus sign will be clear later. For now, note that the weight of a doubled edge is always $+a(v_1,v_2)^2$.  
Lastly, to each isolated vertex $v$, we associate the weight $x(v)$.  The
weight $w(C)$ of a configuration $C$ is then
\be \label{defloopvertexwt}
w(C) = \prod_{\ell \text{ a loop}} w(\ell) \prod_{\substack{v \text{ an } \\ 
\text{isolated vertex}}} \!\! x(v).
\ee

\begin{defn}
The {\em loop-vertex model} on a vertex- and edge-weighted graph $G$ is the collection 
$\mathcal{L}$ of loop-vertex configurations on $G$ with the weight of each configuration given by
\eqref{defloopvertexwt}.
\end{defn}

\begin{figure}[ht]
\setlength{\unitlength}{1mm}
\begin{center}
\begin{picture}(25,25)
\put(0,0){\circle*{1}}
\put(0,10){\circle*{1}}
\put(10,0){\circle*{1}}
\put(10,10){\circle*{1}}
\put(0,20){\circle*{1}}
\put(10,20){\circle*{1}}
\put(20,0){\circle*{1}}
\put(20,20){\circle*{1}}
\put(20,10){\circle*{1}}
\put(-3,-1){1}
\put(-3,9){2}
\put(-3,19){3}
\put(9,-3){4}
\put(11,11){5}
\put(9,21){6}
\put(19,-3){7}
\put(21,10){8}
\put(19,21){9}
\thicklines
\put(0,0){\vector(1,0){10}}
\put(0,0){\vector(0,1){10}}
\put(0,0){\vector(1,1){10}}
\put(10,0){\vector(0,1){10}}
\put(0,10){\vector(1,0){10}}
\put(0,20){\vector(1,0){10}}
\put(10,20){\vector(1,-1){10}}
\put(0,10){\vector(0,1){10}}
\put(0,20){\vector(1,-1){10}}
\put(10,10){\vector(1,0){10}}
\put(10,10){\vector(1,-1){10}}
\put(20,0){\vector(0,1){10}}
\put(10,0){\vector(1,0){10}}
\put(10,0){\vector(1,1){10}}
\put(0,10){\vector(1,-1){10}}
\put(10,20){\vector(1,0){10}}
\end{picture}
\hfil
\begin{picture}(25,25)
\put(0,0){\circle*{1}}
\put(0,10){\circle*{1}}
\put(10,0){\circle*{1}}
\put(10,10){\circle*{1}}
\put(0,20){\circle*{1}}
\put(10,20){\circle*{1}}
\put(20,0){\circle*{1}}
\put(20,20){\circle*{1}}
\put(20,10){\circle*{1}}
\put(-4,-1){1}
\put(-4,9){2}
\put(-4,19){3}
\put(9,-3){4}
\put(11,11){5}
\put(9,21){6}
\put(19,-3){7}
\put(21,10){8}
\put(19,21){9}
\thicklines
\color{blue}
\put(0,10){\vector(0,1){10}}
\put(0,20){\vector(1,-1){10}}
\put(10,10){\vector(1,-1){10}}
\put(20,0){\vector(0,1){10}}
\put(0,0){\circle{3}}
\put(10,20.5){\vector(1,0){10}}
\color{red}
\put(20,10){\vector(-1,-1){10}}
\put(10,0){\vector(-1,1){10}}
\put(20,19.5){\vector(-1,0){10}}
\end{picture}
\end{center}
\caption{A non-planar graph $G$ with its natural orientation on the left. A particular loop-vertex configuration is given on the right where the ``wrongly'' oriented edges are coloured red.
\label{figure.Gexample}}
\end{figure}
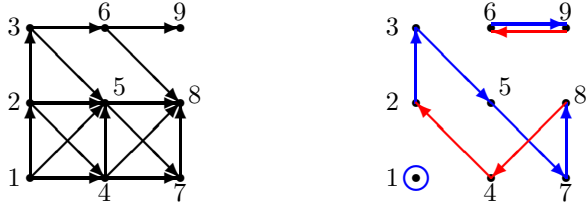

\begin{exmp}
For example, the weight of the configuration in Figure~\ref{figure.Gexample} is 
\[
- x(1) \cdot\; a(6,9)^2  \cdot\;  a(2,3)a(3,5)a(5,7)a(7,8)a(4,8)a(2,4)
\]
\end{exmp}
\noindent
With a slight abuse of terminology, we say that the {\bf (signed) partition function} of the loop-vertex model on the pair $(G,\cO)$ is then
\be \label{defpartfn}
Z_{G,\cO} = \sum_{C \in \mathcal{L}} w(C).
\ee
Whenever the orientation is canonically defined by the labelling on the graph, we will denote the partition function simply as $Z_G$.

\begin{defn}
The {\bf signed adjacency matrix} $K$ associated to the pair $(G,\cO)$ is the matrix $K$ indexed by the vertices of $G$ whose entries are 
\be \label{defkast}
K(v,v') = \begin{cases}
x(v) & v'=v \\
a(v,v') & \text{orientation is from $v$ to $v'$ in $\cO$} \\
-a(v,v') & \text{orientation is from $v'$ to $v$ in $\cO$}.
\end{cases}
\ee
\end{defn}

\begin{thm} \label{thm:pfloopvertex}
The partition function of the loop-vertex model on $(G,\cO)$ is given by
\be \label{partfn-general}
Z_{G,\cO} = \det K.
\ee
\end{thm}

\begin{proof}
We begin by considering the Leibniz formula for the determinant  of $K$. We will 
consider the permutations in $S_n$ according to their cycle decomposition. The
first observation is that the sign of a non-trivial odd cycle
$c=(c_1,c_2,\dots,c_{2l+1},c_1)$ is the opposite of its reverse
$\rev(c)=(c_1,c_{2l+1},\dots,c_2,c_1)$, but the weights are
the same. Therefore, such terms cancel out. The
only odd cycles which appear are cycles of length one, also known as
fixed points.

It is then clear that the terms in the determinant expansion of $K$
are in bijection with loop-vertex configurations of $G$. We now need
to show that the signs are the same.  Therefore we decompose the
permutation $\pi$ into $k$ fixed points and $c$ cycles of lengths
$2m_1,\dots,2m_c$. This ensures that $k$ has the same parity as $n$.

A well-known combinatorial result states that if $n$ is odd
(resp. even), $\pi$ is odd if and only if the number of cycles is even
(resp. odd) in its cycle decomposition. In our case, the number of
cycles is $k+c$. A short tabulation shows that the sign of $\pi$ is
always the same as $(-1)^c$. In other words, the sign of a loop is
precisely the product of all the corresponding terms in $K$ plus one
extra sign. But this is precisely what we have in \eqref{defloop}.
\end{proof}

Although the loop-vertex model consists of signed weights, the following statement can easily be verified since the signed adjacency matrix is a sum of a diagonal matrix and an antisymmetric matrix.

\begin{cor} \label{cor:pf-lvpositive}
The partition function $Z_{G,\cO}$ is a positive polynomial in the variables $x(v)$ for $v \in V$ and $a(v,v')$ for $(v,v') \in E$. In particular, if all the weights are positive reals, $Z_{G,\cO}$ is strictly positive.
\end{cor}

\begin{exmp}
For the loop-vertex model on the complete graph with its canonical orientation (see below \eqref{defpartfn}) with vertex-weights $x$ and edge-weights $a$, the signed adjacency matrix is given by
\[
K_n  = \begin{pmatrix}
x & a & a & \cdots & a \\
-a & x & a & \cdots & a \\
\vdots & \ddots & \ddots & \ddots & \vdots \\
-a & \cdots & -a &  x & a \\
-a & \cdots & -a  & -a & x
\end{pmatrix}
\]
One can compute the determinant of $K_n$ by using elementary row and column operations to convert it to a tridiagonal matrix. It is then easy to show that $Z_n$ satisfies the recursion $Z_n = 2x Z_{n-1} -(x^2-a^2)Z_{n-2}$. 
It is immediate from the initial conditions $Z_0=1$ and $Z_1= x$ that 
\be \label{completegraph}
Z_{n} = \sum_{k=0}^{\lfloor n/2 \rfloor} \binom n{2k} \; x^{n-2k} \; a^{2k} 
 = \frac{(x+a)^n + (x-a)^n}{2}.
 \ee 
\end{exmp}

\section{The Monopole-Dimer Model on Planar Graphs} 
\label{sec:mondim}
We now focus on the model of physical interest, namely the monopole-dimer model. 
As we shall see, technical reasons force us to restrict our attention to planar graphs. 
We will first define the model for an arbitrary planar graph and state a theorem about the partition function of the model.

From now on, we will use $G$ to mean both the graph and its planar
embedding. As before, $G = [V,E]$ will be a labelled graph with vertex weights
$x(v)$ for $v \in V$ and edge weights $a(v,v')$ whenever $(v,v') \in E$.
The configurations of the model are exactly the loop-vertex configurations $\mathcal{L}$ of Definition~\ref{def:loopvertconf}. 
From here on, we will use the term monopole-dimer configurations instead of loop-vertex configurations. 

Let $C \in \mathcal{L}$ be a monopole-dimer configuration containing an even loop $\ell$. The weight of the loop $\ell = (v_1,\dots,v_{2 n},v_1)$ is given by
\be \label{defloop-planar}
w(\ell) = (-1)^{\substack{\text{number of vertices in $V$} \\ 
\text{enclosed by $\ell$}}} \prod_{j=1}^{2n} a(v_j,v_{j+1}),
\ee
where, as before, $v_{2n+1} \equiv v_1$. Notice that the planarity of the graph is used crucially in ensuring that $w(\ell)$ is well-defined. In the usual way, we set the weight of vertex $v$ to be $x(v)$, and the
weight $w(C)$ of the entire configuration $C$ as
\be \label{monopoledimerwt}
w(C) = \prod_{\ell \text{ a loop}} w(\ell) \prod_{v \text{ a vertex}} x(v).
\ee
Note that the definition of the monopole-dimer model on planar graphs is independent
of any orientation, unlike the loop-vertex model.

\begin{defn} \label{def:monopoledimer}
The {\bf monopole-dimer model} on $G$ is a model of monopole-dimer configurations $\mathcal{L}$ on $G$ where the weight of each configuration is given by \eqref{monopoledimerwt}.
\end{defn}

As before, we let the {\bf (signed) partition function} of the monopole-dimer model on $G$ be
\[
Z_{G} = \sum_{\substack{C \text{ a monopole-dimer} \\ \text{configuration}}} w(C).
\]

\begin{rem} \label{rem:doubledimer}
Configurations of the model are superpositions of two configurations of the monomer-dimer model with identical locations of monomers and thus generalise the so-called {\em double-dimer model} \cite{kenyon2011a,kenyon2011b}. Since the weight of each double-dimer loop is given a sign which is the parity of the number of monomers enclosed by it, it is reminiscent of the Dirac string representation of the monopole.
Dirac had shown by integrating the flux around a curve enclosing the string that the well-definedness of the vector potential led naturally to the quantization of charge \cite{dirac1978}. 
\end{rem}

We recall the notion of a {\bf Kasteleyn orientation} for a planar graph. We will consider the case of bipartite graphs for simplicity; the general case is similar. In this case, Kasteleyn \cite{kasteleyn1961} showed
that there exists an orientation $\cO$  on $G$ such that every basic
loop enclosing a face has an odd number of clockwise oriented
edges. This is sometimes called the {\bf clockwise-odd} property.
Using this orientation $\cO$, Kasteleyn showed that the dimer 
partition function on $G$ can be written as a Pfaffian of an even antisymmetric
matrix, now called the Kasteleyn matrix.
Note that the signed adjacency matrix $K$ in \eqref{defkast} differs from the Kasteleyn matrix by a diagonal matrix. In what follows, we will refer to the signed adjacency matrix as a {\bf (modified) Kasteleyn matrix}.

\begin{thm} \label{thm:pfmonopoledimer}
Let $\cO$ be a Kasteleyn orientation on the planar graph $G$ and let $K$ be the modified Kasteleyn matrix defined as \eqref{defkast}. Then the partition function of the monopole-dimer model on $G$ can be written as
\[
Z_G = \det K.
\]
Moreover, Corollary~\ref{cor:pf-lvpositive} immediately implies that $Z_G$ is a positive polynomial in the weights.
\end{thm}

\begin{proof}
To prove this, we have to show that the weight of a loop in a planar graph with a Kasteleyn orientation defined by \eqref{defloop-planar} is the same as that
defined in \eqref{defloop}. Suppose the loop is of length $2 \ell$ and there are $v$ internal vertices, $e$ internal edges and $f$ faces. 
Suppose the Kasteleyn orientation is such that there are $o_j$ clockwise edges in face $j$, where each $o_j$ is odd. The total number of clockwise edges on the loop is therefore $\sum_{j=1}^f o_j - e$ since each internal edge contributes twice to the count, once clockwise and once counter-clockwise. 
The Euler characteristic $v-e+f$ is 1 on the plane since we exclude the unbounded face. Since the parity of $\sum_{j=1}^f o_j - e$ is the same as that of $f-e$, which equals $v-1$, we have shown that the total number of clockwise edges on the loop is odd if and only if $v$ is even. This shows that the weights in \eqref{defloop-planar} and \eqref{defloop} coincide.
\end{proof}
 
\begin{exmp} 
Consider the cycle graph $C_n$ where the vertices
are labelled in cyclic order with the canonical orientation 
with weights  $a$ to each edge
and $x$ to each vertex. The modified Kasteleyn matrix $K_n$ is then
\[
K_n  = \begin{pmatrix}
x & a & 0 & \cdots & a \\
-a & x & a & \cdots & 0 \\
& \ddots & \ddots & \ddots & \\
0 & \cdots & -a & x & a \\
-a & 0 & \cdots & -a & x
\end{pmatrix}.
\]
One can then show with a little bit of work that the partition function $Z_{n}$
satisfies
\[
Z_{n} = \det K_n =
\begin{cases}
x a^n F_{n} \left( \frac xa \right) + 2 a^{n+1} F_{n-1} \left( \frac xa \right), 
& \text{if $n$ is odd}, \\
x a^{n-1} F_{n} \left( \frac xa \right) + 2 a^{n} F_{n-1} \left( \frac xa \right) + 2a^n,
& \text{if $n$ is even},
\end{cases}
\]
where $F_n(x)$ is the $n$'th {\bf Fibonacci polynomial} defined by the recurrence
$F_n(x) = x F_{n-1}(x) + F_{n-2}(x)$ with initial conditions $F_0(x)=0$ and $F_1(x)=1$. 
Using standard properties of the Fibonacci polynomials, we can rewrite
\[
Z_{n} = 
\begin{cases}
\ds a^n \; L_n \left( \frac xa \right) & \text{if $n$ is odd}, \\
\\
\ds a^n  \; \left( L_{n/2} \left( \frac xa \right) \right)^2 & \text{if $n \equiv 0 \;(\bmod\; 4)$ }, \\
\\
\ds a^{n-2}\; (x^2 + 4a^2)\;  \left( F_{n/2} \left( \frac xa \right) \right)^2 
& \text{if $n \equiv 2 \;(\bmod\; 4)$ }.
\end{cases}
\] 
where the {\bf Lucas polynomials} $L_n(x)$ satisfy the same recurrence as the Fibonacci polynomials but with different initial conditions, $L_0(x)=2$ and $L_1(x)=x$.
\end{exmp}

\begin{rem} 
Note that when $n$ is divisible by 4, $\sqrt{Z_n}$ is a positive polynomial and can be considered as the partition function of a model of monomers and dimers. This phenomenon will recur in Section~\ref{sec:mondimgrid}.
\end{rem}

\begin{cor}[Kasteleyn \cite{kasteleyn1961}]
In the absence of vertex weights, i.e. $x(v)=0 \;\forall\; v \in V$, the monopole-dimer model is exactly the double-dimer model (see Remark~\ref{rem:doubledimer}) and consequently, $Z_G = |\pf K|^2$.
\end{cor}

We will now explore some consequences of the determinant formula for the monopole-dimer model. Unlike for the usual dimer model, one cannot calculate probabilities of events for the monopole-dimer model since the measure on configurations here is not positive. However, one can consider expectations of observables in this signed measure.

\begin{defn}
The {\em joint correlation} of a subconfiguration of monopoles $v_1,\dots,v_j$ and loops $\ell_1,\dots,\ell_k$ in the graph $G$ is 
\[
\langle v_1,\dots,v_j; \ell_1,\dots,\ell_k \rangle = \prod_{r=1}^j x(v_r) \prod_{s=1}^k w(\ell_s) \; \frac{\widehat Z_{G'}}{Z_G}
\]
where $G'$ is the subgraph of $G$ with the vertices $v_1,\dots,v_j$ and those in loops 
$\ell_1,\dots,\ell_k$ removed; and $\widehat Z_{G'}$ is the partition function of the monopole-dimer model in $G'$ with the caveat that the sign of loops in $G'$ given by 
\eqref{defloop-planar} be taken by considering vertices in all of $G$
\end{defn}

To give a formula for joint correlations, we recall the complementary minor identity of Jacobi. For a $k \times k$ nonsingular matrix $A$, $1 \leq i \leq k$, sequences $[p] = (p_1,\dots,p_i), [q] = (q_1,\dots,q_i)$ where 
$1 \leq p_1 < \cdots < p_i \leq k$ and $1 \leq q_1 < \cdots < q_i \leq k$, let $A^{[p]}_{[q]}$ be the $i \times i$ submatrix of $A$ consisting of rows $p_j$ and columns $q_j$. Also, let $[\bar{p}]$ (resp. $[\bar{q}]$) be the complementary sets $\{1,\dots,k\} \setminus [p]$ (resp. $\{1,\dots,k\} \setminus [q]$). Recall that the determinant of such a submatrix is called a {\em minor}, and when $[p]=[q]$, both the submatrix and its minor are qualified by the adjective {\em principal}.

\begin{thm}[Jacobi, see $\S 14.16$ of \cite{gradry2000}] \label{thm:jacobi}
\[
\det \left( (A^{-1})^{[\bar q]}_{[\bar p]} \right)= \frac{(-1)^{p_1+ q_1 +\cdots+ p_i+q_i}}{\det A}
\det A^{[p]}_{[q]}.
\]
\end{thm}

\begin{rem} \label{rem:nonnegwts}
In the case of an {\em unsigned} combinatorial model on a graph (i.e. with nonnegative weights) whose partition function can be written as a determinant, Theorem~\ref{thm:jacobi} implies that probabilities of local events can be computed in terms of principal minors of the inverse. This fact has been used with great success for the dimer model \cite{kenyon1997}.
\end{rem}

\begin{lem} \label{lem:jointcor}
Let $[\bar p]$ be the set of vertices of monopoles $v_1,\dots,v_j$ and of loops $\ell_1,\dots,\ell_k$. Then the joint correlation is given by
\[
\langle v_1,\dots,v_j; \ell_1,\dots,\ell_k \rangle = \prod_{r=1}^j x(v_r) \prod_{s=1}^k w(\ell_s) \; \det \left((K^{-1})^{[\bar p]}_{[\bar p]} \right).
\]
Furthermore, it is positive if the subconfiguration has positive weight.
\end{lem}

\begin{proof}
The sum over all configurations with these prescribed monopoles and loops is given by the appropriate principal minor of the modified Kasteleyn matrix, $K$. By Theorem~\ref{thm:jacobi}, this is exactly the complementary minor of the inverse, which exists because of Corollary~\ref{cor:pf-lvpositive}. The minor is the determinant of  a matrix which is the sum of an antisymmetric matrix and a diagonal matrix and  is positive, again using Corollary~\ref{cor:pf-lvpositive}. Thus, the only way for the joint correlation to be negative is if the subconfiguration itself has negative weight.
\end{proof}

We will use Lemma~\ref{lem:jointcor} in the following to compute monopole correlations.
The joint correlation of the monopole-loop configuration on the right side of Figure~\ref{figure.squares3} in a large $m \times n$ square is the simplest example of one which is negative. 

\section{The Monopole-Dimer model on the Rectangular Grid} 
\label{sec:mondimgrid}
We will now calculate the partition function for the monopole-dimer model in Section~\ref{sec:mondim} on the 
rectangular grid graph, thereby
generalising the famous product formula of Temperley-Fisher \cite{fisher1961} and
Kasteleyn \cite{kasteleyn1961}. 

Consider the $m \times n$ grid $Q_{m,n} = \{(i,j) \;|\; 1 \leq i \leq m, 1 \leq j \leq n\}$ with  
horizontal edge-weights $a$, vertical edge-weights $b$ and vertex-weights $z$.
For the sake of completeness, we recall the Kasteleyn orientation $\cO$ prescribed independently by Fisher \cite{fisher1961} and Kasteleyn \cite{kasteleyn1961}.  
The arrow always points in the direction $(i,j) \to (i,j+1)$, i.e., towards the positive $x$-axis. In the $y$-direction, the arrow points from $(i,j) \to (i,j+1)$ (i.e. towards the positive $y$-axis) whenever $i$ is odd and in the reverse direction when $i$ is even. This orientation can be easily seen to be induced by a ``snake-like'' labelling, as seen in Figure~\ref{figure.squares3}.

Define the function
\[
Y_m(b;z) = \prod_{j=1}^{\lfloor m/2 \rfloor} \left( z^2 + 4 b^2 \cos^2
 \frac{j \pi}{m+1}  \right).
\]

\begin{thm} \label{thm:partfngrid}
The partition function of the monopole-dimer model on $Q_{m,n}$ is given by
\begin{equation}	
\begin{split} \label{partfn-grid}
Z_{{m,n}}=& \prod_{j=1}^{\lfloor m/2 \rfloor}  \prod_{k=1}^{\lfloor n/2 \rfloor}  
\left( z^2 + 4 b^2 \cos^2  \frac{j \pi}{m+1}  + 4 a^2 \cos^2  \frac{k \pi}{n+1} \right)^2 \\
&\times \begin{cases}
1 & \text{if $m$ and $n$ are even}, \\
Y_m(b;z) & \text{if $m$ is even and $n$ is odd}, \\
Y_n(a;z) & \text{if $m$ is odd and $n$ is even}, \\
z Y_m(b;z) Y_n(a;z) & \text{if $m$ and $n$ are odd}.
\end{cases}
\end{split}
\end{equation} 
\end{thm}

\begin{proof}
The matrix $K_{m,n}$ is exactly the regular Kasteleyn matrix added to $z$ times the identity matrix of size $mn$. Therefore, the inversion technique described in either of these papers works identically when $m$ or $n$ are even. 
The case when both $m$ and $n$ are odd is a special case, which has to be worked out separately. Both cases can be analysed simultaneously.

We use Fisher's labelling \cite{fisher1961}. The modified Kasteleyn matrix can be written in $m \times m$ tridiagonal block form 
\[
K = \begin{pmatrix}
X & Y & 0 & 0 &\cdots & 0 \\
-Y & X & Y & 0 &\cdots & 0 \\
0 & -Y & X & Y & \cdots & 0 \\
\vdots & \ddots & \ddots & \ddots & \ddots & \vdots \\
0 & \cdots & 0 & -Y & X & Y \\
0 & 0 & \cdots & 0 & -Y & X
\end{pmatrix},
\]
where each of the blocks is an $n \times n$ matrix with
\[
X = \begin{pmatrix}
z & a & 0 & 0 & \cdots & 0 \\
-a & z & a & 0 & \cdots & 0 \\
0 & -a & z & a  & \cdots & 0 \\
\vdots & \ddots & \ddots & \ddots & \ddots & \vdots \\
0 & \cdots & 0 & -a & z & a \\
0 & 0 & \cdots & 0 & -a & z
\end{pmatrix}, \quad
\text{ and }
Y = \begin{pmatrix}
0 & \cdots & 0 & 0 & b \\
0 & \cdots & 0 & b & 0 \\
\vdots & \Ddots & \Ddots & \Ddots & \vdots \\
0 & b & 0 & \cdots & 0 \\
b & 0 & 0 & \cdots & 0
\end{pmatrix}.
\]
Fisher showed \cite{fisher1961} that the matrix $K$ can be simplified considerably by the unitary transformation $U = u_m \otimes u_n$ where $u_s$ is an $s \times s$ matrix with entries
\be \label{defu}
(u_s)_{p,q} = \sqrt{\frac{2}{s+1}} i^p \sin \frac{\pi p q }{s+1},
\ee
where $i = \sqrt{-1}$. $u_s^{-1}$ also has an equally simple formula
\be \label{defuinv}
(u_s)^{-1}_{p,q} = \sqrt{\frac{2}{s+1}} (-i)^q \sin \frac{\pi p q }{s+1}.
\ee
One can show that $U^{-1}\, K\, U$ can be written as a block diagonal matrix of
$\chi_s$'s for $s=1,\dots,m$, where
\be \label{defchi}
\chi_s = \begin{pmatrix}
\ss z + 2ia \cos \frac{\pi}{n+1} & \ss 0 & \dots & \ss 0 & \ss (-1)^{n-1} 2 i^n b \cos \frac{\pi s}{m+1} \\
\ss 0 & \ss z + 2ia \cos \frac{2\pi}{n+1} & & \ss (-1)^{n-2} 2 i^n b \cos \frac{\pi s}{m+1} & \ss 0 \\
\vspace{-0.2cm}\vdots & &  \ddots \Ddots & & \vdots \\
\vdots & &  \Ddots \ddots & & \vdots \\
\ss 0 & \ss (-1)^{1} 2 i^n b \cos \frac{\pi s}{m+1} & & \ss z + 2ia \cos \frac{(n-1)\pi}{n+1} & \ss 0 \\
\ss (-1)^{0} 2 i^n b \cos \frac{\pi s}{m+1} & \ss 0 & \dots & \ss 0 & \ss z + 2ia \cos \frac{n \pi}{n+1}
\end{pmatrix}
\ee
Let us now look at each of the cases. When $n$ is even, the determinant of $\chi_s$ is easily expressed as a product of $2 \times 2$ determinants,
\begin{align*}
\det\chi_s &= \prod_{p=1}^{\frac{n}{2}} 
\begin{vmatrix}
\ds z + 2ia \cos  \frac{p \pi}{n+1} &  (-1)^{n-p} 2 i^n b \cos  \frac{s \pi}{m+1} \\
\\
\ds (-1)^{p-1} 2 i^n b \cos  \frac{s \pi}{m+1}  & z + 2ia \cos \frac{(n+1-p) \pi}{n+1} 
\end{vmatrix} \\
&= \prod_{p=1}^{\frac{n}{2} } 
\left( z^2 + 4 b^2 \cos^2  \frac{s \pi}{m+1} + 4 a^2 \cos^2  \frac{p \pi}{n+1}  \right)
= \det\chi_{m+1-s}.
\end{align*}
When $m$ is also even, we get after multiplying over all $s$, precisely the formula in the first case \eqref{partfn-grid}.
When $m$ is odd,  we get an additional factor $|\chi_{(m+1)/2}|$, which is easy to compute because it is a diagonal matrix. The factor we get is
\begin{align*}
\det\chi_{(m+1)/2} &= \prod_{p=1}^{n } \left( z + 2 i a  \cos  \frac{p \pi}{n+1} \right) \\
&= \prod_{p=1}^{\frac{n}{2} } \left( z^2 + 4 a^2  \cos^2  \frac{p \pi}{n+1} \right) = Y_n(a;z).
\end{align*}
This also matches with \eqref{partfn-grid}.
When $n$ is odd and $m$ is even, we have the additional factor contributing to each $|\chi_s|$ from the central term,
\[
z+ 2ib \cos \frac{\pi s}{m+1} .
\]
Multiplying this factor for all $s$ gives us $Y_m(b;z)$ as needed. The last case when both $m$ and $n$ are odd gives us both the factors above and an additional term corresponding to the central entry of the central block matrix, which can be seen to be $z$. 
\end{proof}

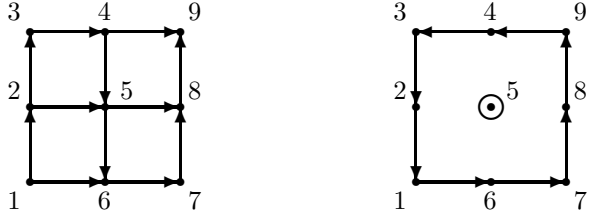
\begin{figure}[ht]
\setlength{\unitlength}{1mm}
\begin{center}
\begin{picture}(25,25)
\put(0,0){\circle*{1}}
\put(0,10){\circle*{1}}
\put(10,0){\circle*{1}}
\put(10,10){\circle*{1}}
\put(0,20){\circle*{1}}
\put(10,20){\circle*{1}}
\put(20,0){\circle*{1}}
\put(20,20){\circle*{1}}
\put(20,10){\circle*{1}}
\put(-3,-3.5){1}
\put(-3,11){2}
\put(-3,21.5){3}
\put(9,-3.5){6}
\put(12,11){5}
\put(9,21.5){4}
\put(21,-3.5){7}
\put(21,11){8}
\put(21,21.5){9}
\thicklines
\put(0,0){\vector(1,0){10}}
\put(0,0){\vector(0,1){10}}
\put(10,10){\vector(0,-1){10}}
\put(10,0){\vector(1,0){10}}
\put(20,0){\vector(0,1){10}}
\put(10,10){\vector(1,0){10}}
\put(0,10){\vector(0,1){10}}
\put(0,10){\vector(1,0){10}}
\put(10,20){\vector(0,-1){10}}
\put(20,10){\vector(0,1){10}}
\put(0,20){\vector(1,0){10}}
\put(10,20){\vector(1,0){10}}
\end{picture}
\hfil
\begin{picture}(25,25)
\put(0,0){\circle*{1}}
\put(0,10){\circle*{1}}
\put(10,0){\circle*{1}}
\put(10,10){\circle*{1}}
\put(0,20){\circle*{1}}
\put(10,20){\circle*{1}}
\put(20,0){\circle*{1}}
\put(20,20){\circle*{1}}
\put(20,10){\circle*{1}}
\put(-3,-3.5){1}
\put(-3,11){2}
\put(-3,21.5){3}
\put(9,-3.5){6}
\put(12,11){5}
\put(9,21.5){4}
\put(21,-3.5){7}
\put(21,11){8}
\put(21,21.5){9}
\thicklines
\put(0,0){\vector(1,0){10}}
\put(10,0){\vector(1,0){10}}
\put(20,0){\vector(0,1){10}}
\put(20,10){\vector(0,1){10}}
\put(20,20){\vector(-1,0){10}}
\put(10,20){\vector(-1,0){10}}
\put(0,20){\vector(0,-1){10}}
\put(0,10){\vector(0,-1){10}}
\put(10,10){\circle{3}}
\end{picture}
\end{center}
\caption{The graph $Q_{3,3}$ of Example~\ref{ex:q33} with its natural orientation on the left and a particular monopole-dimer configuration on the right.
\label{figure.squares3}}
\end{figure}

\begin{exmp} \label{ex:q33}
Consider the first nontrivial case: $m=n=3$. 
Figure~\ref{figure.squares3} shows  $Q_{3,3}$ with a Kasteleyn orientation and one of two configurations which contribute with a weight $-z a^4 b^4$.
The modified Kasteleyn matrix is given by
\[
K_{{3,3}} = \left(
\begin{array}{ccccccccc}
 z & a & 0 & 0 & 0 & b & 0 & 0 & 0 \\
 -a & z & a & 0 & b & 0 & 0 & 0 & 0 \\
 0 & -a & z & b & 0 & 0 & 0 & 0 & 0 \\
 0 & 0 & -b & z & a & 0 & 0 & 0 & b \\
 0 & -b & 0 & -a & z & a & 0 & b & 0 \\
 -b & 0 & 0 & 0 & -a & z & b & 0 & 0 \\
 0 & 0 & 0 & 0 & 0 & -b & z & a & 0 \\
 0 & 0 & 0 & 0 & -b & 0 & -a & z & a \\
 0 & 0 & 0 & -b & 0 & 0 & 0 & -a & z \\
\end{array}
\right)
\]
\vspace{0.3cm}
The partition function is given by
\[
Z_{{3,3}} =z \left(2 a^2+z^2\right) \left(2 b^2+z^2\right) \left(2 a^2+2 b^2+z^2\right)^2,
\]
in agreement with \eqref{partfn-grid}.
\end{exmp}

\begin{rem}
The fact that the partition function $Z_{m,n}$ is an exact square of a positive
polynomial when $m$ and $n$ are even is nontrivial since $K_{m,n}$ is not antisymmetric. 
\end{rem}

We will now use Theorem~\ref{thm:jacobi} to calculate joint correlations in the monopole-dimer model on the grid graph. We focus on the case when $m,n$ are even for simplicity. To do so, we will first need to calculate the matrix entries for the inverse of the modified Kasteleyn matrix. We will now calculate this in full generality.
Since we have used the snake-like Kasteleyn orientation explained at the beginning of this Section, the relationship between the entries of the matrix and coordinates on the grid depends on the parity of the abscissa. To simplify notation, we define the functions
\begin{align*}
\phi_{g,h}(c,d;e,f) = \frac{4 i^{c+d} (-i)^{e+f}}{(m+1)(n+1)} 
&\sin \frac{\pi c g}{m+1} \sin \frac{\pi e g}{m+1} 
\sin \frac{\pi d h}{n+1} \sin \frac{\pi f h}{n+1}
\\
& \times\left( \frac{z - 2 i a \cos \frac{\pi h}{n+1} +(-1)^{f+h} 2i^n b \cos \frac{\pi g}{m+1}}
{z^2  + 4 a^2 \cos^2 \frac{\pi h}{n+1} + 4 b^2 \cos^2 \frac{\pi g}{m+1}}
\right), \\
\psi_{g,h}(c,d;e,f) = \frac{4 i^{c+d} (-i)^{e+f}}{(m+1)(n+1)} 
&\sin \frac{\pi c g}{m+1} \sin \frac{\pi e g}{m+1} 
\sin \frac{\pi d h}{n+1} \sin \frac{\pi f h}{n+1}
\\
& \times\left( \frac{z - 2 i a \cos \frac{\pi h}{n+1} + (-1)^{f+h-1} 2i^n b \cos \frac{\pi g}{m+1}}
{z^2  + 4 a^2 \cos^2 \frac{\pi h}{n+1} + 4 b^2 \cos^2 \frac{\pi g}{m+1}}
\right),
\end{align*}
for fixed $m,n$ and weights $a,b,z$. The only difference between the two functions is the power of $-1$ in the last term of the numerator inside the parenthesis.

\begin{lem} \label{lem:Kinv}
If $n$ is even, the entries of the inverse matrix are given by
\[
\left(K_{m,n}^{-1} \right)_{(c,d),(e,f)} = \sum_{(g,h) \in Q_{m,n}}
\begin{cases}
\phi_{g,h}(c,d;e,f) & \text{if both $c$ and $e$ are odd,} \\
\phi_{g,h}(c,n+1-d;e,f) & \text{if $c$ is even and $e$ is odd,} \\
\phi_{g,h}(c,d;e,n+1-f) & \text{if $c$ is odd and $e$ is even,} \\
\phi_{g,h}(c,n+1-d;e,n+1-f) & \text{if both $c$ and $e$ are even,} \\
\end{cases}
\]
and if $n$ is odd, the entries are given by
\[
\left(K_{m,n}^{-1} \right)_{(c,d),(e,f)} = \sum_{(g,h) \in Q_{m,n}}
\begin{cases}
\psi_{g,h}(c,d;e,f) & \text{if both $c$ and $e$ are odd,} \\
\psi_{g,h}(c,n+1-d;e,f) & \text{if $c$ is even and $e$ is odd,} \\
\psi_{g,h}(c,d;e,n+1-f) & \text{if $c$ is odd and $e$ is even,} \\
\psi_{g,h}(c,n+1-d;e,n+1-f) & \text{if both $c$ and $e$ are even.} \\
\end{cases}
\]
\end{lem}

\begin{proof}
Since $U^{-1}\, K_{m,n} \, U = \text{Diag}(\chi_1,\dots,\chi_m)$, where $\chi_s$ is given in \eqref{defchi} and $U = u_m \otimes u_n$, $U^{-1}$ are given explicitly in \eqref{defu},\eqref{defuinv},
one starts by inverting $\chi_s$ and obtains $K_{m,n}^{-1}$ as $U^{-1} \, \text{Diag}(\chi_1^{-1},\dots,\chi_m^{-1}) \, U$ by a somewhat lengthy but straightforward calculation. 

The parities of $c$ and $e$ enter in the calculation simply 
because in the Kasteleyn orientation, the coordinates $d,f$ increase from left to right when $c,e$ are odd and from right to left, when $c,e$ are even; see Figure~\ref{figure.squares3} for example.
\end{proof}

\begin{cor} \label{cor:1pt}
The one-point monopole correlation (informally the density) at $(c,d)$ in the $m \times n$ grid is given by
\[
\frac{4 z^2}{(m+1)(n+1)} \sum_{(g,h) \in Q_{m,n}}
\!\! \frac{\sin^2 \frac{\pi c g}{m+1} \sin^2 \frac{\pi d h}{n+1}}
{z^2  + 4 a^2 \cos^2 \frac{\pi h}{n+1} + 4 b^2 \cos^2 \frac{\pi g}{m+1}}.
\]
\end{cor}

\begin{proof}
As per the definition of the correlation and Theorem~\ref{thm:jacobi}, the one-point monopole correlation at $(c,d)$ is given by $z\; K^{-1}_{(c,d),(c,d)}$. We use Lemma~\ref{lem:Kinv} and use the symmetry of $\phi_{g,h}(c,d;c,d)$ and $\psi_{g,h}(c,d;c,d)$ under the transformations $g \mapsto m+1-g$ and $h \mapsto n+1-h$ to obtain the result. 
\end{proof}

\begin{figure}[htp!]
\begin{center}
\includegraphics[height=6cm]{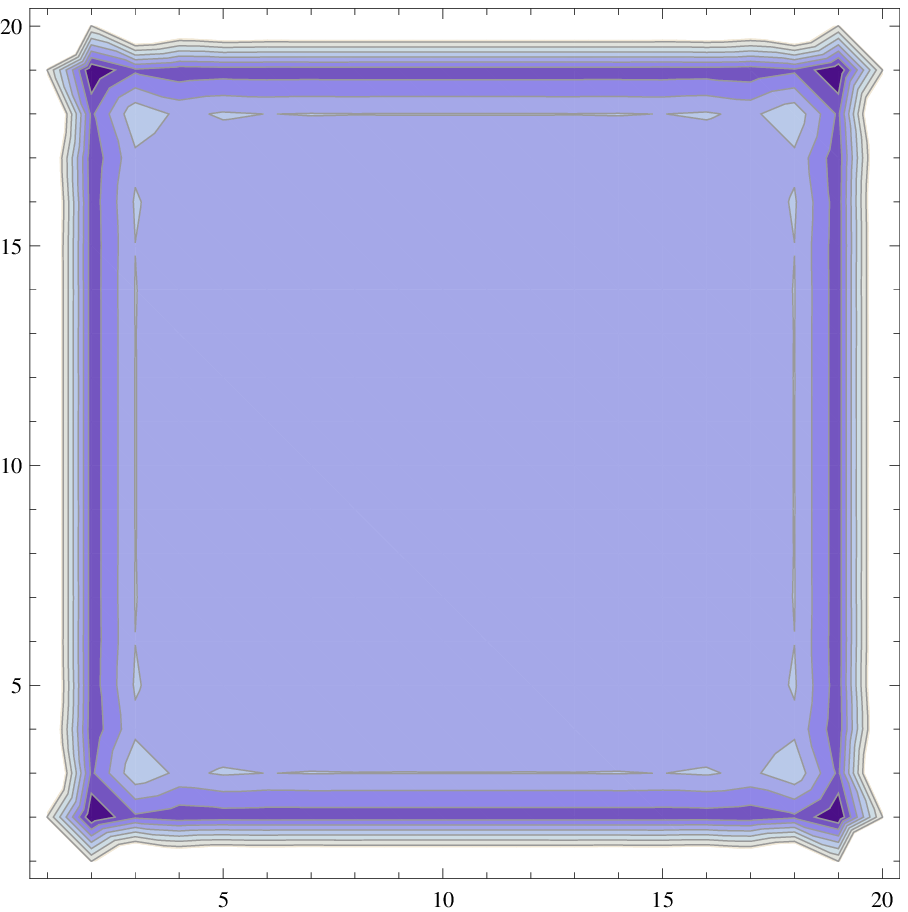}
\hspace{0.2cm}
\raisebox{1cm}{\includegraphics[height=4cm]{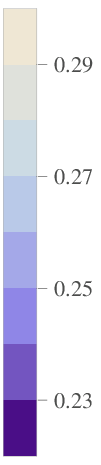}}
\caption{A contour plot of exact monopole densities as a function of location in a $20 \times 20$ grid for $a=b=z=1$. Note the uniformity of the density in the interior of the system.
\label{fig:mondens}}
\end{center}
\end{figure}

As expected, the prefactor of $z^2$ ensures that there is one additional monopole when either $m$ or $n$ is even. 
See, for example, Figure~\ref{fig:mondens} for the density plot when $m=n=20$. 
One can also compute joint correlations of monopoles. For instance, the two-point correlation of monopoles at positions $(c,d)$ and $(e,f)$ far apart is given by
\[
\det \left( \begin{matrix}
K^{-1}_{(c,d),(c,d)} & 
K^{-1}_{(c,d),(e,f)}\\
K^{-1}_{(e,f),(c,d)} & 
K^{-1}_{(e,f),(e,f)}
\end{matrix} \right).
\]

\section{Discussion on Asymptotic Behaviour} 
\label{sec:asymp}
We will now focus on asymptotic results for the monopole-dimer model on the grid graph.
The results in this section will be less formal and will focus more on obtaining rough estimates for the equivalent of quantities in standard thermodynamics, such as the free energy, density and the entropy.	Part of the reason for the informality of this section is that we are manipulating $Z_{m,n}$ as if it were the standard partition function in statistical physics. This is not strictly allowed because our partition functions are signed sums. However, as we shall see, we can justify this {\em a posteriori} by showing that the results are sensible.

Just as at the end of the previous section, $m,n$ are assumed to be even for simplicity. 
The {\bf free energy} is then given by
\[
F(a,b,z) = \lim_{m,n \to \infty} \frac 1{mn} \ln Z_{{m,n}}.
\]
Using \eqref{partfn-grid}, one can treat the right hand side as a Riemann sum, which tends to the limit
\begin{align*}
F(a,b,z) = \frac{2}{\pi^2} \int_0^{\pi/2} \text{d} \theta & \int_0^{\pi/2} 
\text{d} \phi \; \ln \left( z^2 + 4 a^2 \cos^2 \theta + 4 b^2 \cos^2 \phi \right).
\end{align*}
Following standard thermodynamic relations, the density of $a$-type of dimers (and similarly, the $b$-type) and that of monopoles is given, after differentiating under the integral sign, by
\begin{align}
\rho_{a} &= a \frac{\partial}{\partial a} F(a,b,z) = \frac{2}{\pi^2} \int_0^{\pi/2} \text{d} \theta & \int_0^{\pi/2} 
\text{d} \phi \; \frac{8a^2 \cos^2 \theta}{z^2 + 4 a^2 \cos^2 \theta + 4 b^2 \cos^2 \phi}
, \\
\rho_z &= z \frac{\partial}{\partial z} F(a,b,z) = \frac{2}{\pi^2} \int_0^{\pi/2} \text{d} \theta & \int_0^{\pi/2} 
\text{d} \phi \; \frac{2z^2}{z^2 + 4 a^2 \cos^2 \theta + 4 b^2 \cos^2 \phi}.
\end{align}
It is easy to see that
\(
\rho_a + \rho_b + \rho_z = 1.
\)
This is to be expected since each vertex either contains a monopole or is part of a loop adjacent to either an $a$ or a $b$ dimer. 

One of the integrals in each case is easily done. Surprisingly, $\rho_a$ is easier to evaluate than $\rho_z$ even though the final formula will turn out to be simpler for the latter. 
\[
\rho_a = \frac{2}{\pi} \int_0^{\pi/2} \text{d} \theta 
\frac{4a^2 \cos^2 \theta}{\sqrt{(z^2 + 4 a^2 \cos^2 \theta )
(z^2 + 4 b^2 + 4 a^2 \cos^2 \theta )}}
\]
After the change of variables $t = (2a \cos \theta)^{-1}$, we obtain
\[
\rho_a = \frac{1}{\pi a y z} \int_{1/2a}^{\infty} \frac{\text{d}t}
{t\sqrt{(t^2 - \frac{1}{4 a^2})(t^2 + \frac{1}{z^2})(t^2 + \frac{1}{y^2})}},
\]
where $y^2 = z^2 + 4 b^2$. Using a known formula for elliptic integrals \cite{gradry2000}[(3.137), Formula 8] and an amazing transformation \cite{absteg1964}[Formula 17.7.14], it turns out that $\rho_a$ can be concisely expressed in terms of a single known special function, {\em the Heuman Lambda function} $\Lambda_0(\theta,k)$, defined in \cite[Formula 17.4.39]{absteg1964}, as
\be \label{rhoa}
\rho_a =  1- \Lambda_0 (\theta_a,k),
\ee
where all the complexity has been absorbed in the parameters
\be
\theta_a = \tan^{-1} \left( \sqrt{\frac{4b^2+z^2}{4a^2}} \right), \; \text{ and }
k = \frac{4ab}{\sqrt{(4a^2+z^2)(4b^2+z^2)}},
\ee
where $k$ is the standard notation for the elliptic modulus.
We remark that the Heuman Lambda function is an elliptic function related to the Jacobi Zeta function and has come up in various physical problems.
It turns out that the monopole density can be written, using a miraculous addition formula for $\Lambda$ \cite[Formula 153.01]{byrd1971} as
\be \label{mondens}
\rho_z = \frac{K(k)\; k \; z^2}{2 \, \pi \, a \, b },
\ee
where $K(k)$ is the complete elliptic integral of the first kind.
Now that we have expressions for $\rho_a$ and $\rho_z$, we would like to obtain a 
simple expression for the free energy $F(a,b,z)$ by integrating $\frac{\rho_z}z$. 

Starting with the series expansion for $K(k)$ in \cite[Formula 17.3.11]{absteg1964}, we get
\be \label{rhoz-sum}
\rho_z = \frac{z^2}{\sqrt{(4a^2+z^2)(4b^2+z^2)}} \sum_{j=0}^\infty 
\left(\frac{(2j-1)!!}{2^j j!}\right)^2 \left( \frac{16 a^2 b^2 }{(4a^2+z^2)(4b^2+z^2)}
\right)^j.
\ee
We now integrate $\frac{\rho_z}{z}$ term by term assuming $a > b$ and obtain, using a standard computer algebra package, an infinite sum involving $2F1-$hypergeometric functions, 
\be \label{fe-a>b}
\begin{split}
F(a,b,z) = \sqrt{\frac{4b^2+z^2}{a^2-b^2}} \; & \sum_{j=0}^\infty
\; \frac{1}{2j-1} \;\binom{2j}j^2
\\ & \times
\left( \frac{a^2b^2}{(a^2-b^2)(4b^2+z^2)} \right)^{j}
\pFq{2}{1}{\frac 12-j,\frac 12+j}{\frac 32-j}{\frac{4b^2+z^2}{4b^2-4a^2} }
\end{split}
\ee
Since the integrands are symmetric in $a$ and $b$, we can obtain the free energy when $b>a$ by interchanging $a$ and $b$ in \eqref{fe-a>b}.
We handle the $a=b$ case separately. 
In particular, we set them equal to 1 without loss of generality.
In that case, each integral in \eqref{rhoz-sum} is easier because of the absence of square roots and it turns out that we can write $F(1,1,z)$ again using a computer algebra package as
\be\label{fe-ab=1}
F(1,1,z) = \frac{1}{2} \ln(4 + z^2) - \frac{1}{(4+z^2)^2} \; 
\pFq{4}{3}{1,1,\frac 32,\frac 32}{2,2,2}{\frac{16}{(4+z^2)^2}}
\ee

\begin{figure}[htp!]
\begin{center}
\includegraphics[height=5cm]{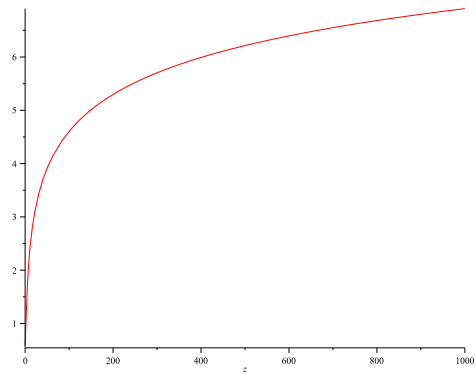}
\caption{A plot of $F(1,1,z)$ for $z$ varying between 0 and 1000.
\label{fig:fe}}
\end{center}
\end{figure}

As expected from general statistical physical considerations, 
$F(1,1,z)$ grows monotonically in $z$  and is concave 
\cite{hammersley1966} as seen in Figure~\ref{fig:fe}.
In accordance with the intuition developed for the monomer-dimer model \cite{gruberkunz1971,heilmann1972}, $F(1,1,z)$ is smooth and there are no phase transitions.
One can verify that $F(1,1,0) = 2G/\pi$, where $G$ is Catalan's constant. This is expected since this model reduces, when $z=0$, to the double-dimer model, which is the square of the dimer model \cite{kasteleyn1961,fisher1961}.

Since configurations of the monopole-dimer model are superpositions of two dimer model configurations with fixed monopole locations, and since $Z_{m,n}$ turns out to be a perfect square, we can compare $\sqrt{Z_{m,n}}$ with existing literature on the monomer-dimer model.
$F(1,1,z)/2$ compares favourably with rigorous bounds for the free energy of the classical monomer-dimer model in the literature, although it is not very close to numerical data; see Figure~\ref{fig:fe-comparison}. Note that the plot here is as a function of dimer density $\rho = \rho_a + \rho_b$, not $z$. The transformation is a classic exercise in demonstrating equivalence of ensembles. See \cite[Appendix A]{kong2006b} for example.

\begin{figure}[htp!]
\begin{center}
\includegraphics[height=5cm]{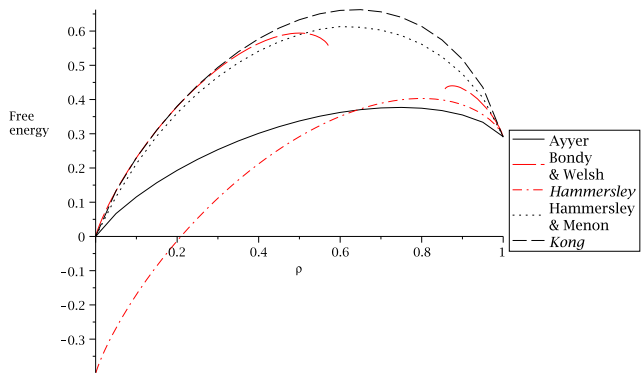}
\caption{Comparison of $F(1,1,z)/2$ with bounds obtained in \cite{bondy1966,hammersley1966,hammersley1970} and recent numerics in \cite{kong2006b}.
\label{fig:fe-comparison}}
\end{center}
\end{figure}

One can also calculate the {\bf entropy} using standard thermodynamic relations,
\begin{align*}
S(a,b,z) &= F(a,b,z) - z \ln z \frac{\partial}{\partial z} F(a,b,z)\\
&= F(a,b,z) - \rho_z \ln z .
\end{align*}
In the special case of equal dimer weights, this leads to
\[
S(1,1,z) = F(1,1,z) - \frac{2z^2 \ln z}{\pi(4+z^2)} 
K \left( \frac{4}{4+z^2} \right).
\]
Using \eqref{fe-ab=1}, one can show that the entropy is maximum when $z=1$, at which point the monopole density using \eqref{mondens} is 
\[
\rho_z = \frac{2}{5 \pi} K \left(\frac{4}{5} \right) \approx 0.25404.
\]
This compares very well with the exact result for the $20 \times 20$ grid in Figure~\ref{fig:mondens}.

Many qualitative properties of the monopole-dimer model on grids are similar to those of the classical monomer-dimer model, which is of much interest to scientists in various fields. The exact formulas for grids presented here might be used to gain further insight about the monomer-dimer model. 
The determinantal character of the partition function for the loop-vertex model on general graphs and the monopole-dimer model on planar graphs might also prove useful in other contexts.

\bibliographystyle{alpha}
\bibliography{dimermodels}

\end{document}